\documentclass[letterpaper,11pt]{article}
\usepackage[utf8]{inputenc}
\usepackage{latexsym}
\usepackage{amssymb}
\usepackage[margin=1in]{geometry}
\usepackage{amsmath,amsfonts,amsthm}
\usepackage{dsfont}
\usepackage{graphicx}
\usepackage{multirow}
\usepackage{ifsym}
\usepackage{color}
\usepackage{hyperref}
\usepackage{dsfont}
\usepackage{xcolor}
\usepackage{algorithm}
\usepackage{algorithmicx}
\usepackage{algpseudocode}
\usepackage[numbers]{natbib}
\hypersetup{%
   breaklinks,%
   colorlinks=true,%
   linkcolor=black,%
   urlcolor=black,%
   citecolor=[rgb]{0,0,0.45}
}

\theoremstyle{plain}

\theoremstyle{plain}
\newtheorem{theorem}{Theorem}[section]
\newtheorem{lemma}[theorem]{Lemma}

\newtheorem{definition}[theorem]{Definition}
\theoremstyle{definition}
\newtheorem{observation}[theorem]{Observation}
\newtheorem{remark}[theorem]{Remark}
\newtheorem{corollary}[theorem]{Corollary}

\newtheorem*{restatementinner}{\restatementname}

\newenvironment{restatement}[2]
  {\def\restatementname{#1~\ref{#2}}%
   \begin{restatementinner}}
  {\end{restatementinner}}

\DeclareSymbolFont{bbold}{U}{bbold}{m}{n}
\DeclareSymbolFontAlphabet{\mathbbold}{bbold}

\newcommand{\nats}{{\mathbb N}}

\newcommand{\eps}{{\varepsilon}}

\newcommand{\opt}{\mathrm{OPT}}

\newcommand*{\doub}[1]{\overline{\overline{#1}}}

\title{An $n^{0.3+\eps}$-Approximation for Steiner $k$-Forest}
\author{Eden Chlamt\'a\v{c}\thanks{Institute of Theoretical Computer Science, Ben-Gurion University.}}
\date{}

\begin{document}

\maketitle
\begin{abstract}
    We give an $n^{0.3+\varepsilon}$-approximation algorithm for the Steiner $k$-Forest problem, for any constant $\varepsilon>0$. As a function of $n$, this improves over the $O(\min\{\sqrt{n},\sqrt{k}\})$-approximation of Gupta et al.\ [ESA'07, TALG'10] which has stood for nearly two decades for the general case, as well as the later $n^{0.448}$-approximation of Dinitz et.\ al [APPROX-RANDOM'14, TALG'17] for the uniform weight case.

    On the other hand, we show that, due to lower bounds on the Densest $k$-Subgraph problem, the $O(\sqrt k)$-approximation for Steiner $k$-Forest likely cannot be improved. Specifically, we show that for any sufficiently small $\eps>0$, 
    an $O(k^{1/2-\eps})$-approximation for Steiner $k$-Forest would surpass known degree-$n^{\Omega(\eps^2)}$ Sum-of-Squares integrality gaps for Densest k-Subgraph and refute the corresponding dense-versus-random conjecture.
\end{abstract}

\section{Introduction}

In the Steiner $k$-Forest problem, we are given an undirected graph $G=(V,E)$ with non-negative edge weights $w$, sets $S,T\subseteq V$ of vertices which we call terminals, a set of terminal pairs $P\subseteq S\times T$, and a parameter $k\in[|P|]$. The goal is to find a minimum weight subgraph in which at least $k$ pairs $(s,t)\in P$ are connected.

Unlike its non-parametrized counterpart, Steiner Forest, which has a 1.994-approximation~\cite{GT26}, and its parametrized single-source demand version $k$-MST, which has a 2-approximation~\cite{garg05}, Steiner $k$-Forest seems genuinely harder. Even the simple case of a unit-weight star graph is already equivalent to Smallest $k$-Edge Subgraph (as shown in~\cite{HJ06}), where the goal is to find a subgraph with at least $k$ edges and a minimum number of vertices. For this problem, as for its better known maximization version, Densest $k$-Subgraph, there is compelling evidence for an $n^{\Omega(1)}$ lower bound~\cite{chlamtavc2018sherali,DkSSoS}.

For $n$-vertex graphs and terminal connectivity parameter $k$, the best known approximation for Steiner $k$-Forest prior to our work is an $O(\min\{\sqrt{n},\sqrt{k}\})$-approximation due to Gupta et al.~\cite{GHNR10}. For the special case of unit weight edges, Dinitz et al.~\cite{DKN17} were able to obtain an $n^{0.448}$-approximation, and neither the unit weight or the general weight case has been improved since then. On the other hand, the special case of Smallest $k$-Edge Subgraph admits a much better approximation: $n^{3-2\sqrt{2}+\eps}\approx n^{0.17+\eps}$ for any $\eps>0$~\cite{CDK12}, which raises the possibility that Steiner $k$-Forest itself may admit a much better approximation. Indeed, we show the following (in the general case, for arbitrary graphs and arbitrary weights):

\begin{theorem}
    For any constant $\eps>0$, there is a randomized polynomial time $n^{0.3+\eps}$-approximation algorithm for Steiner $k$-Forest.
    \label{thm:main}
\end{theorem}

It is natural to wonder, since the $O(\sqrt{n})$-approximation can be significantly improved, whether a similar improvement shouldn't also hold for the $O(\sqrt{k})$-approximation. However, we observe that, for an appropriate parameter regime, powerful SoS integrality gaps (and a corresponding hardness conjecture that predates them) for Densest $k$-Subgraph translate to $k^{1/2-\eps}$ integrality gaps 
for Smallest $k$-Edge Subgraph (for any constant $\eps>0$), making a significant improvement to the $O(\sqrt{k})$-approximation for Steiner $k$-Forest evidently unlikely.

Finally, we note that Gupta et al.~\cite{GHNR10} also gave a reduction that is nearly approximation preserving (up to an $O(\log ^2n)$-factor) from the Dial-a-Ride problem\footnote{Here we are given an $n$-point metric, a collection of items each with a source and destination, and need to plot a minimum length route to get every item from its source to its destination using a vehicle that can carry at most $k$ items at a time.} to Steiner $k$-Forest. As a consequence, our algorithm also translates to an $n^{0.3+\eps}$-approximation for Dial-a-Ride.

\subsection{Overview of our results}
To greatly simplify the problem, we use tree embeddings, a common tool for network connectivity problems (see, e.g.,~\cite{AA97,GKR9800, cccdggl9899,CEK06,CEGS0811}). As a further simplification, we also employ a known height reduction embedding in exchange for an $O(n^{\eps})$ loss in the approximation guarantee (the source of the $n^{\eps}$ factor in our result).\footnote{We note that with a more careful analysis, $\eps$ could be reduced from ``any constant" to some $o_n(1)$ but we make no attempt to optimize the $n^{\eps}$ factor here, as we believe this would distract from the main ideas in our algorithm.}

Beyond this preprocessing, our improvement is essentially an extension of the $O(\min\{\sqrt{n},\sqrt{k}\})$-approximation of~\cite{GHNR10}. At a high level, the $O(\sqrt{n})$ part of their guarantee can be thought of as a tradeoff between two algorithms: an $O(\hat n)$-approximation, and an $O(n/\hat n)$-approximation, where~$\hat n$ is essentially the number of distinct terminals involved in an optimum solution. We observe that the second of these has a much stronger analog in the star graph case (that is, Smallest $k$-Edge Subgraph): a greedy $O(\sqrt{n/\hat n})$-approximation. Extending this algorithm to our more general setting brings the tradeoff between that and the $O(\hat n)$-approximation down to $\tilde O(n^{1/3})$. Finally, introducing one more new algorithmic idea and combining with the $O(\sqrt{k})$-approximation of~\cite{GHNR10} further improves the tradeoff between all the various algorithms to $\tilde O(n^{3/10})$, which, after factoring in the loss from the embedding, gives us $n^{0.3+\eps}$.

While the greedy $O(\sqrt{n/\hat n})$-approximation in the star case is quite simple and intuitive, it is far from immediate that it can be extended to more elaborate tree structures. To achieve this, we first apply additional preprocessing to make the input tree locally approximately regular, and for such instances, design a general procedure that translates edge-sampling algorithms that achieve good terminal-pair connectivity in expectation to ones that achieve good connectivity with high probability. It is then a matter of very carefully choosing probability parameters at different levels in the tree to match the desired guarantee.

\paragraph{Organization of the paper.} In Section~\ref{sec:preliminaries}, we describe the various preprocessing steps, introduce our notation, and give a very simple one-sided tree algorithm. In Section~\ref{sec:subsampling}, we give our general procedure for edge-sampling algorithms with high success probability. In Section~\ref{sec:sqrt} we present our main algorithm (the analog of the $O(\sqrt{n/\hat n})$-approximation), which suffices to get an $n^{1/3+\eps}$ guarantee. In Section~\ref{sec:improve} we give our final improvement, obtaining $n^{0.3+\eps}$, and finally in Section~\ref{sec:hard}, we elaborate on the special case of Smallest $k$-Edge Subgraph and the implied conditional hardness of improving over $O(\sqrt{k})$.

\section{Preliminaries}\label{sec:preliminaries}
We begin by considering the parametrized ratio version of the problem:

\begin{definition}
    In the Minimum-Ratio $k$-Tree (MR$k$T) problem, given an input $\langle G,w,S,T,P,k\rangle$ as for Steiner $k$-Forest, find a tree $H$ connecting at least $\hat k$ terminal pairs in $P$ for some $\hat k\in[k]$, and minimizing the ratio $w(H)/\hat k$.
\end{definition}
There is a simple reduction from Steiner $k$-Forest to this problem: repeatedly run an approximation algorithm for MR$k$T (each time updating the number of pairs left to satisfy), and contract the tree found by the algorithm to a single vertex. A standard argument shows that this reduction translates any $\tilde O(n^{\alpha})$-approximation for MR$k$T to an $\tilde O(n^{\alpha})$-approximation for Steiner $k$-Forest. Thus, in the rest of the paper, even when not explicitly specified, the algorithm we design and analyze is an approximation algorithm for MR$k$T.

Next, using the FRT embedding into a distribution on trees~\cite{FRT04}, we may assume that the input graph is a tree, and lose at most an $O(\log n)$-factor in the approximation ratio (this can be done either before or after the reduction to MR$k$T). 

With both of these reductions in mind, we next guess a vertex $r\in V$ which participates in an optimum solution to MR$k$T (by trying every possibility), consider it as the root of our input tree (after FRT), and perform the height reduction embedding guaranteed by the following lemma, (Claim 2.4 in~\cite{CEK06}, slightly rephrased):

\begin{lemma}
    For some universal constant $C>1$, there exists a linear time algorithm that, given a rooted tree $G_1=(V_1,E_1)$ with weights $w_1$ and root $r_1$, and a parameter $\alpha\geq C$, returns a rooted tree $G_2=(V_2,E_2)$ with weights $w_2$, root $r_2$, and an embedding $\pi:V_1\to V_2$, with the following properties:
    \begin{enumerate}
        \item For any vertex set $U\subseteq V_1$, let $H_1$ be the subtree of $G_1$ spanning $\{r_1\}\cup U$ and $H_2$ be the subtree of $G_2$ spanning $\{r_2\}\cup \pi(U)$. Then $$w_2(H_2)\leq \alpha\cdot w_1(H_1).$$
        \item For any vertex set $U\subseteq V_2$, let $H_2$ be the subtree of $G_2$ spanning $\{r_2\}\cup U$ and $H_1$ be the subtree of $G_1$ spanning $\{r_1\}\cup\pi^{-1}(U)$. Then $$w_1(H_1)\leq w_2(H_2).$$
        \item The height of $G_2$ is at most $O(\log_{\alpha} n)$.
    \end{enumerate}
    \label{lem:CEK}
\end{lemma}

For any small $\eps>0$, plugging in $\alpha=n^{\eps}$, we get a tree of height $O(1/\eps)$. Since we assume the vertex~$r$ is known to be part of our optimum MR$k$T solution, our algorithm will consist of choosing a set of  terminals and returning the subtree connecting these terminals to $r$, and thus by Lemma~\ref{lem:CEK} we may assume our input tree has depth at most $O(1/\eps)$, while paying an additional $n^{\eps}$-factor in the approximation.

Having established our final embedding, we also guess the parameters of an optimum solution to MR$k$T: an approximate bound $\opt$ on the total edge weight of the solution (up to a constant factor), and the number of pairs satisfied in the solution, which we denote by $\hat k$. Note that w.l.o.g., edge weights fall in a somewhat narrow range: we can delete every edge of weight greater than $\opt$, and contract every edge of weight at most $\opt/n$ (since the input graph is a tree and we are aiming for a super-constant approximation ratio anyway).

To simplify the algorithm, we perform a regularization procedure which guarantees layer-wise near-regularity for the input tree and the terminal pairs. 
\begin{lemma}
    Given a bounded-diameter $n$-vertex tree $G=(V,E)$ with positive edge weights $w$ in some range $[w_0,n\cdot w_0]$, a root $r\in V$, and a set of terminal pairs $P\subseteq S\times T$ where $S,T\subseteq V$, 
    we can efficiently find terminal subsets $S'\subseteq S$ and $T'\subseteq T$ and a set of pairs $P'\subseteq P\cap(S'\times T')$ of size $|P'|=\tilde\Omega(|P|)$, with the following properties:
    \begin{itemize}
    \item For some $c=O(1)$ that depends only on the diameter of $G$, there exist $D_S,D_T\in\nats$ such that every terminal $s\in S'$ has terminal-pair degree in the range $D_S/\log^c n\leq\deg_{P'}(s)\leq D_S$, and every terminal $t\in T'$ has degree in the range $D_T/\log^c n\leq\deg_{P'}(t)\leq D_T$.    
    \item There exist $g,h\in\nats$ such that all terminals in $S'$ are at depth $g$ and all terminals in $T'$ are at depth $h$ (note that these are constants bounded by the diameter of $G$).
    \item There exist $b_1,\ldots,b_g\in\nats$ such that, in the rooted subtree spanning $\{r\}\cup S'$, at each layer $i\in[g]$, every parent node has fan-out in the range $[b_i/\log^c n,b_i]$. Moreover, every edge in this layer has weight in the range $[\beta_i,2\beta_i]$ for some $\beta_i> 0$.
    \item There exist $d_1,\ldots,d_h\in\nats$ such that, in the rooted subtree spanning $\{r\}\cup T'$, at each layer $j\in[h]$, every parent node has fan-out in the range $[d_j/\log^c n,d_j]$. Moreover, every edge in this layer has weight in the range $[\delta_j,2\delta_j]$ for some $\delta_j> 0$.
    \item There exist $\tilde b_1,\ldots,\tilde b_g$ such that for every terminal $t\in T'$, in the rooted subtree spanning $\{r\}\cup \Gamma_{P'}(t)$, at each layer $i$ (of the $S'$-side of the tree), every parent node has fan-out in the range $[\tilde b_i/\log^c n,\tilde b_i]$.
    \item There exist $\tilde d_1,\ldots,\tilde d_h$ such that for every terminal $s\in S'$, in the rooted subtree spanning $\{r\}\cup \Gamma_{P'}(s)$, at each layer $j$ (of the $T'$-side of the tree), every parent node has fan-out in the range $[\tilde d_j/\log^c n,\tilde d_j]$.
\end{itemize}
    \label{lem:reg}
\end{lemma}

This regularizing procedure is largely based on a regularizing procedure of Chlamt\'a\v{c} and Dinitz~\cite{Chlamtac-Dinitz}, and given in Appendix~\ref{app:reg} for completeness. Note that the subtree spanning $\{r\}\cup S'$ and the subtree spanning $\{r\}\cup T'$ may overlap in the actual graph, but we will treat them as coming from two separate copies of the graph joined at $r$ (and thus pay at most an additional factor $2$). For later consistency, all the ranges for degree parameters in the above lemma should be read as giving tight upper bounds. That is, ``all degrees in this set are in the range $[x,y]$" here means ``all these degrees are in this range and at least one is exactly $y$."

Repeatedly finding a set $P'\subseteq P$ as guaranteed by Lemma~\ref{lem:reg} and deleting these pairs from~$P$, partitions the pairs $P$ into $\tilde O(1)$ many buckets, one of which must contain $\tilde\Omega(\hat k)$ pairs in our optimum solution. Thus, we will run our algorithm on each of these buckets separately. For now, let $(P',S',T')$ be a bucket that contains $\tilde\Omega(\hat k)$ such pairs.

Applying the same procedure (in theory) to our optimum solution would yield a similarly regularized solution connecting at least $\tilde\Omega({\hat k})$ pairs, with total edge weight still at most $\opt$. We will focus on such an optimum solution in our algorithm and analysis.

Notation: After applying Lemma~\ref{lem:reg} to both the input and the optimum solution, we use the parameters $D_S,D_T,b_i,d_j,\tilde b_i,\tilde d_j$ defined in the lemma to refer to the parameters of the regularized set of paths $P'$ in the input. 
Note that the number of terminals in $S'$ is $$n_S:=|S'|=\tilde\Theta(1)\cdot \prod_{i=1}^gb_i,$$ and the number of terminals in $T'$ is 
$$n_T:=|T'|=\tilde\Theta(1)\cdot\prod_{j=1}^h d_j.$$ 
For simplicity, we use a hat accent to indicate the corresponding parameters in our optimum solution, slightly adjusted so that degrees are aligned to the lower end of each range rather than the upper end. For example, the fan-out for each parent in layer two on the $S'$ side of the optimum tree is in the range $[\hat b_2,\hat b_2\cdot \log^c n]$ (with at least one having fan-out $\hat b_2$). This ensures monotonicity of solution versus input parameters (e.g., $\hat D_S\leq D_S$). Note that we thus have $$\hat k=\tilde\Theta(\hat n_S \hat D_S)=\tilde\Theta(\hat n_T \hat D_T).$$ For the purpose of our algorithm, we guess all the parameters of the regularized optimum solution.

We make one additional simplification: By our assumption regarding the parameters in the optimum solution, every layer of the optimum tree has cost at most $\opt$. We increase (or very slightly decrease) the current weights as necessary, so the cost becomes $\tilde\Theta(\opt)$. So, for example, as there are $\tilde\Theta(1)\cdot\prod_{i=1}^{g'}\hat{b}_i$ edges in layer $g'$ of the $S'$-side optimum tree, the current weight bound $\beta_{g'}$ must be at most $\tilde O(\opt/\prod_{i=1}^{g'}\hat{b}_i)$ by definition. Setting all edge weights in this layer to exactly $\opt/\prod_{i=1}^{g'}\hat{b}_i$ gives completely uniform weights within each layer, without affecting the approximation guarantee by more than a polylogarithmic factor (and we adjust weights on the $T'$-side similarly, as well).

\subsection{One-sided tree and connection to $k$-MST}

We will occasionally need to find a rooted subtree to connect just one side (say, the $S$ side) to $r$ after we've already fixed a set of terminals on the other side. This can easily be reduced to $k$-MST, but since the input graph is a tree, this kind of optimization can also be done exactly rather than relying on approximation algorithms for more general graphs. More precisely, we have the following (the roles of $S$ and $T$ can of course be exchanged):

\begin{lemma}
    There is a polynomial time exact algorithm for the following problem: Given a weighted tree graph with root $r$, terminal pairs $P\subseteq S\times T_0$, 
    and a parameter $k'$, find a set of terminals $S_0\subseteq S$ such that
    \begin{enumerate}
        \item The number of satisfied demand pairs is bounded by $|P\cap(S_0\times T_0)|\geq k'$, and
        \item The rooted subtree spanning $\{r\}\cup S_0$ has minimum cost.
    \end{enumerate}
    \label{lem:one-side}
\end{lemma}
\begin{proof}[Proof sketch:]
    This can be accomplished easily via dynamic programming. At every node $u$, for every~$k''$, we calculate the cost of the minimum subtree rooted at $u$ that satisfies $\geq k''$ demand pairs (if there is such a subtree, otherwise, $\infty$). A single leaf terminal $s\in S$ has cost $0$ for all $0\leq k''\leq |P\cap(\{s\}\times T_0)|$. For any non-leaf $u$, we order its children, and use dynamic programming to calculate these values for trees involving its first $j$ children, adding one child at a time, finally combining the information for its children with any demand pairs $u$ itself can satisfy in case $u$ is a terminal in $\Gamma_P(T_0)$.
\end{proof}

We note that, in our regularized setting, correctly guessing a terminal $t\in T'$ that participates in the regularized optimum solution, we can buy the path from $t$ to $r$ at cost at most $\opt$ (and likely significantly less), and then use Lemma~\ref{lem:one-side} with $T_0=\{t\}$ and $k'=\hat D_T$ to find a cost $O(\opt)$ tree connecting $\hat D_T=\tilde\Theta(\hat k/\hat n_T)$ terminal pairs. This is precisely the basis of the $O(\hat n)$-approximation algorithm of~\cite{GHNR10} mentioned earlier. Thus we have the following:

\begin{corollary}
    In the context of Lemma~\ref{lem:reg} (and the following discussion of its application to the optimum solution), there is an $\tilde O(\min\{\hat n_S,\hat n_T\})$-approximation algorithm for MR$k$T.
    \label{cor:hat-n}
\end{corollary}

\section{A general subsampling algorithm}\label{sec:subsampling}

Two of our algorithms rely on the following idea: on the $S'$ side, in each layer $i\in[g]$ retain every edge independently with probability $p_i$, and on the $T'$ side, in each layer $j\in[h]$ retain every edge independently with probability $q_j$. In expectation, the number of terminal pairs in $P'$ that remain connected is approximately $n_SD_S\prod_{i=1}^g p_i\prod_{j=1}^h q_j$, and we can similarly upper bound the expected cost of the tree. However, since sampling one side of the tree (say, the $S'$ side) may completely change the structure of the relevant $T'$ side tree, we may not get concentration around the expectation (or a high probability of success) if the $T'$ side is subsampled independently. Thus, we need a slightly more subtle subsampling procedure, which we describe in Algorithm~\ref{alg:subsample}.

\begin{algorithm}[h]
    \caption{Subsampling procedure}
    \label{alg:subsample}
\begin{algorithmic}
\State \textbf{Input:} A regularized tree as in Lemma~\ref{lem:reg}, probabilities $p_i\in[1/b_i,1]$ and $q_j\in[1/d_j,1]$ for all $i\in[g],j\in[h]$.
    \For{all $i=1$ to $g$} \Comment{Sample the $S'$ side}
       \For{every edge $e$ in layer $i$ on the $S'$ side of the tree}
       \State{\textbf{retain} edge $e$ with probability $\min\{1,p_i\log^{c+2} n$\}} independently. 
       \EndFor
    \EndFor
    \State{\textbf{let} $S''$ be the set of terminals in $S'$ still connected to $r$}
    \State{\textbf{call} the regularization procedure from Lemma~\ref{lem:reg} for the $T'$-side tree spanning $r\cup\bigcup_{s\in S''}\Gamma_{P'(s)}$ and pairs $P'\cap (S''\times T')$.}
    \State{\textbf{let} $T''$ be the set of remaining terminals in $T'$ after the regularization and $P''$ be the set of remaining connected terminal pairs.}
    \State{\textbf{let} the new fan-out parameters on the $T''$ side be $\bar d_1,\ldots,\bar d_h$} 
    \For{all $j=1$ to $h$} \Comment{Sample the $T''$ side}
       \For{every edge $e$ in layer $j$ on the $T''$ side of the tree}
       \State{\textbf{retain} edge $e$ with probability $\min\{1,(q_jd_j/\bar d_j)\log^{c+2} n$\}} independently. 
       \EndFor
    \EndFor
    \State{\textbf{return} the subtree spanning all terminals still reachable from $r$}
\end{algorithmic}
\end{algorithm}

We show that our subsampling procedure behaves just as if we sampled edges independently as in our na\"{i}ve approach (in terms of expected cost and terminal connectivity), but with a high probability of achieving roughly the expected number of connected terminal pairs.

\begin{lemma} Given a regularized instance satisfying the guarantees in Lemma~\ref{lem:reg}, and probability parameters $p_i,q_j$ satisfying $1/b_i\leq p_i\leq 1$ and $1/d_j\leq q_j\leq 1$ for every $i\in[g],j\in[h]$, Algorithm~\ref{alg:subsample} returns a subtree with the following properties:
\begin{enumerate}
    \item With high probability, the number of terminal pairs in $P'$ that remain connected is at least $$\tilde\Omega\left(n_SD_S\cdot \prod_{i=1}^g p_i\cdot \prod_{j=1}^{h}q_j\right).$$
    \item The expected cost of the subtree is $$\opt\cdot\tilde O\left(\max_{g'\in[g]}\prod_{i=1}^{g'}\frac{p_ib_i}{\hat b_i}+\max_{h'\in[h]}\prod_{j=1}^{h'}\frac{q_jd_j}{\hat d_j}\right).$$
\end{enumerate}
\label{lem:subsampling}
\end{lemma}
\begin{proof}
    We start with the bound on the number of connected terminal pairs. Let us first analyze the subsampling on the $S'$ side of the tree. We can easily bound the number of surviving paths from $r$ up through each layer inductively. Indeed, at each layer $i$, since every parent node $u$ has degree in the range $[b_i/\log^c n, b_i]$, and since $p_i\in[1/b_i,1]$, by a Chernoff bound, the number of children it retains after subsampling will be tightly concentrated around $\min\{p_i\log^{c+2} n,1\}\deg(u)=\tilde\Theta(p_ib_i)$. Thus after this subsampling phase, we get our set of connected terminals $S''$ with
    $$\bar n_S :=|S''|= \tilde\Theta\left(\prod_{i=1}^{g}p_ib_i\right)=\tilde\Theta\left(n_S\cdot\prod_{i=1}^g p_i\right).$$ 

    Note that every terminal $s\in S''$ is still involved in $\tilde\Theta(D_S)$ terminal pairs with terminals in~$T'$, thus giving a total of $\tilde\Theta(\bar n_S D_S)$ terminal pairs in $P''$ still connected after the first subsampling phase. An $\tilde\Omega(1)$-fraction of these pairs are retained after the additional regularization step, thus giving us $\bar n_T:=|T''|$ terminals of remaining terminal-pair degree $\tilde\Theta(\bar D_T)$ such that $\bar n_T\bar D_T=\tilde\Theta(\bar n_SD_S)$.

    Note that the second subsampling phase (on the reduced $T'$ side) is defined exactly as the first, except with new probability parameters $\bar q_j=\min\{1,q_j d_j/\bar d_j\}$. Since we clearly have $\bar q_j\in[1/\bar d_j,1]$, the same concentration argument as before applies to the number of terminals pairs in $T''$ that survive the second subsampling phase. With high probability, up to polylogarithmic factors, it will be at least
    \begin{align*}
        \prod_{j=1}^h \bar q_j \bar d_j &= \prod_{j=1}^h \min\{q_jd_j,\bar d_j\}\\
        &\geq \prod_{j=1}^h q_j\bar d_j&\text{since }\bar d_j\leq d_j\text{ and }q_j\leq 1\\
        &=\tilde\Theta\left(\bar n_T\cdot\prod_{j=1}^h q_j\right)\\
        &=\tilde\Theta\left(\frac{\bar n_S D_S}{\bar D_T}\cdot\prod_{j=1}^h q_j\right)&\text{since }\bar n_T\bar D_T=\tilde\Theta(\bar n_S D_S)
    \end{align*}

Finally, as each of these terminals is involved in $\tilde\Theta(\bar D_T)$ connected terminal pairs, we get the desired lower bound on the total number of connected terminal pairs, which, up to polylogarithmic factors, is at least
$$\bar n_S D_S\cdot\prod_{j=1}^h q_j=\tilde\Theta(1)\cdot n_S D_S\cdot\prod_{i=1}^g p_i\cdot \prod_{j=1}^h q_j.$$ (Note that this is a lower bound. The concentration result holds for a possibly larger expectation.)

We now turn to bounding the expected cost of the solution. This follows more easily by linearity of expectation. Start with the $S'$ side. Recall that in layer $g'$ we have $\tilde\Theta(\prod_{i=1}^{g'}b_i)$ edges each with weight $\opt/\prod_{i=1}^{g'}\hat b_i$. Thus, the expected weight of edges in this layer that remain connected to $r$ is, up to polylogarithmic factors
$$\frac{\opt}{\prod_{i=1}^{g'}\hat b_i}\cdot \prod_{i=1}^{g'}p_i\cdot\prod_{i=1}^{g'}b_i,$$ which sums up to the left side of our stated bound on the cost. The cost on the $T$ side is calculated similarly, with layer $h'$ giving expected weight approximately equal to
$$\opt\cdot\prod_{j=1}^{h'}\frac{\bar q_j\bar d_j}{\hat d_j}\leq \opt\cdot\prod_{j=1}^{h'}\frac{q_jd_j}{\hat d_j},$$ where the inequality follows since we defined $\bar q_j=\min\{1,q_j d_j/\bar d_j\}\leq q_j d_j/\bar d_j$. The above bound, in turn, sums up to the right side of our stated cost bound, which completes the proof.
\end{proof}

Note that, by standard arguments, both the cost upper bound and the connectivity lower bound in Lemma~\ref{lem:subsampling} can be guaranteed simultaneously with high probability via Markov's inequality for the upper bound and repeated trials.

\section{Main algorithm}\label{sec:sqrt}

\subsection{A na\"ive bound}

Consider the outcome of applying the subsampling procedure of Algorithm~\ref{alg:subsample} with probability parameters $p_i=\hat b_i/b_i$ and $q_j=\hat d_j/d_j$. These clearly satisfy the required upper and lower bounds $1/b_i\leq p_i\leq 1$ and $1/d_j\leq q_j\leq 1$. And so, by Lemma~\ref{lem:subsampling}, we get a solution of expected cost $\tilde O(\opt)$ in which w.h.p.\ the number of connected demand pairs is, up to polylogarithmic factors, at least
$$n_SD_S\cdot\prod_{i=1}^{g}\frac{\hat b_i}{b_i}\cdot\prod_{j=1}^{h}\frac{\hat d_j}{d_j}=\tilde\Theta\left(n_SD_S\cdot\frac{\hat n_S}{n_S}\cdot\frac{\hat n_T}{n_T}\right)=\tilde\Omega\left(D_S\hat n_S\cdot\frac{\hat n_T}{n_T}\right).$$ By our assumption, this must be at most ${\hat k}=\hat n_S\hat D_S$ -- otherwise, we could just use this solution and get an $\tilde O(1)$-approximation to MR$k$T. Thus, we get\footnote{This bound is motivated by optimality, and could be enforced with significantly more careful preprocessing. The alternative $\tilde O(1)$-approximation gives a simpler way to arrive at the same bound.}

\begin{observation} We may henceforth assume that $$\hat n_T D_S\leq n_T\hat D_S.$$
\label{obs:naive-bound}
\end{observation}


\subsection{Algorithm for the large solution size regime}
We describe here the algorithm for the large $\min\{\hat n_S,\hat n_T\}$ case. It is best understood as a generalization of the following simple greedy algorithm for the special case of a star graph (when both $S'$ and $T'$ are at depth $1$), phrased here in the context of our regularized tree: Choose a set $S''$ of $\sqrt{n_T\hat n_S\hat D_T/D_S}$ terminals in $S'$. They contribute $\tilde\Theta(1)\cdot\sqrt{n_T\hat n_S\hat D_T D_S}$ terminal pairs. Out of the $n_T$ terminals in $T'$, let $T''$ be the $\sqrt{n_T\hat n_T\hat D_S/D_S}$ terminals with the most terminal-neighbors in $S''$. The number of terminal pairs they contribute is at least a $|T''|/|T'|$ fraction of the total, which comes out to $\tilde\Theta(1)\cdot\sqrt{n_T\hat n_S\hat D_T D_S}\cdot\sqrt{n_T\hat n_T\hat D_S/D_S}\cdot 1/n_T=\tilde\Theta(\hat k)$. On the other hand, the edges on the $S'$ side cost $\opt/\hat n_S$, and on the $T'$ side $\opt/\hat n_T$, and so the cost of this solution is $\opt\cdot\tilde\Theta(1)\cdot\sqrt{n_T\hat D_T/(\hat n_S D_S)}$ on the $S'$ side, and $\opt\cdot\tilde\Theta(1)\cdot\sqrt{n_T\hat D_S/(\hat n_T D_S)}$ (which is the same), thus giving an approximation ratio of $\tilde O\left(\sqrt{n_T\hat D_S/(\hat n_T D_S)}\right)$.

To get the same guarantee, both in terms of the number of connected terminal pairs and overall cost, in a regularized tree of arbitrary (but constant) depth, we will apply our subsampling procedure with carefully chosen probabilities at every layer. We describe this in Algorithm~\ref{alg:sqrt}.

\begin{algorithm}[h!]
    \caption{Generalized greedy algorithm}
    \label{alg:sqrt}
\begin{algorithmic}
\State \textbf{Input:} A regularized tree as in Lemma~\ref{lem:reg}.
    \State{\textbf{let}} \begin{equation}
        g_0:=\max\left\{0\leq g'\leq g\;\left|\ \prod_{i=1}^{g'}\frac{\hat b_i}{b_i}\geq\sqrt{\frac{D_S}{\hat D_S}\cdot\frac{\hat n_T}{n_T}}\right.\right\}.
        \label{eq:g0-def}
    \end{equation}
    \State{\textbf{let}} \begin{equation}
        h_0:=\max\left\{0\leq h'\leq h\;\left|\ \prod_{j=1}^{h'}\frac{\hat d_j}{d_j}\geq\sqrt{\frac{D_S}{\hat D_S}\cdot\frac{\hat n_T}{n_T}}\right.\right\}.
        \label{eq:h0-def}
    \end{equation}
    \If{$g_0=g$ (respectively, $h_0=h$)}
       \State{\textbf{add}} the entire subtree spanning $\{r\}\cup S'$ (resp.\ $\{r\}\cup T'$) 
       to the solution.
       \State{\textbf{call} the one-sided procedure from Lemma~\ref{lem:one-side} with $S_0=S'$ (resp.\  $T_0=T'$) and $k'=\tilde\Omega({\hat k})$.}
       \State{\textbf{add} the resulting subtree on the $T$-side (resp.\  $S$-side) to the solution.}
       \State{\textbf{return} the combined subtree.}
    \Else
       \State{\textbf{call} Algorithm~\ref{alg:subsample}} 
       with probability parameters $p_1,\ldots,p_g$ and $q_1,\ldots, q_h$, defined as follows:
       \begin{align*}
           &p_i=1&\forall i\leq g_0\\
           &p_{g_0+1}=\sqrt{\frac{\hat D_S}{D_S}\cdot\frac{n_T}{\hat n_T}}\cdot\prod_{i=1}^{g_0+1}\frac{\hat b_i}{b_i}\\
           &p_i=\frac{\hat b_i}{b_i}&\forall i\geq g_0+2\\
           &q_j=1&\forall j\leq h_0\\
           &q_{h_0+1}=\sqrt{\frac{\hat D_S}{D_S}\cdot\frac{n_T}{\hat n_T}}\cdot\prod_{j=1}^{h_0+1}\frac{\hat d_j}{d_j}\\
           &q_j=\frac{\hat d_j}{d_j}&\forall j\geq h_0+2
       \end{align*}
       \State{\textbf{return} the resulting subtree.}
    \EndIf
\end{algorithmic}
\end{algorithm}

\begin{lemma}
    In our setting, algorithm~\ref{alg:sqrt} gives an $\tilde O\left(\sqrt{n_T\hat D_S/(\hat n_T D_S)}\right)$-approximation for MR$k$T.
    \label{lem:sqrt}
\end{lemma}
\begin{proof}
    We first note that the layers where we can afford to buy all the edges are precisely those indexed by $i\leq g_0$ and $j\leq h_0$. Indeed, by Observation~\ref{obs:naive-bound}, the expression on the right hand side of the inequalities in both~\eqref{eq:g0-def} and~\eqref{eq:h0-def} is at most 1, so these inequalities hold at least for the empty product (i.e., for $g'=0$, $h'=0$). Also note that the product on the left hand side is non-increasing with $g'$ (resp.\  $h'$), so the inequality holds for $g'$ (resp.\  $h'$) iff $g'\leq g_0$ (resp.\  $h'\leq h_0$). Since every edge in layer $g'$ has weight $\opt/\prod_{i=1}^{g'}\hat b_i$, and the number of edges in this layer is $\tilde\Theta(\prod_{i=1}^{g'}\hat b_i)$, for any $g'\leq g_0$, the cost of layer $g'$ on the $S'$ side is bounded, up to polylogarithmic factors, by
    $$\opt\cdot\prod_{i=1}^{g'}\frac{b_i}{\hat b_i}\leq \opt\cdot\sqrt{\frac{\hat D_S}{D_S}\cdot\frac{n_T}{\hat n_T}},$$ which is consistent with the desired approximation guarantee for this algorithm assuming $\tilde\Omega({\hat k})$ demand pairs are connected. A similar bound holds for layers $h'\leq h_0$ on the $T'$ side. 
    
    Thus, if $g_0=g$, then buying the entire subtree spanning $\{r\}\cup S'$ costs us no more than we are willing to pay, and by the assumption about the optimum solution, the additional one-sided tree procedure from Lemma~\ref{lem:one-side} on the $T$-side buys us a subtree of additional cost at most $\opt$ which together with the $S'$-side tree connects at least $\tilde\Omega({\hat k})$ pairs. The case where $h_0=h$ is symmetric.

    Now consider the case where $g_0<g$ and $h_0<h$ (in particular, they could be $0$). We first note that for the parameters $(p_i)_{i\in[g]}$ and $(q_j)_{j\in[h]}$ defined in Algorithm~\ref{alg:sqrt}, the bounds $1/b_i\leq p_i\leq 1$ and $1/d_j\leq q_j\leq 1$ required for Algorithm~\ref{alg:subsample} (see Lemma~\ref{lem:subsampling}) clearly hold at least for all $i\in[g]\setminus\{g_0+1\}$ and $j\in[h]\setminus\{h_0+1\}$. However, they may not hold for $p_{g_0+1}$ and $q_{h_0+1}$. To justify the use of Algorithm~\ref{alg:subsample} with these parameters, note that Algorithm~\ref{alg:sqrt} is equivalent to adding all edges in $S'$-side layers $i\leq g_0$ and $T'$-side layers $j\leq h_0$, contracting these layers, and then calling the subsampling procedure Algorithm~\ref{alg:subsample} on the partially contracted tree. In this tree, what was originally $S'$-side layer $g_0+1$ is now a layer with a single parent ($r$) and fan-out $\tilde\Theta(\mathbf{b})$, where $\mathbf{b}=\prod_{i=1}^{g_0+1}b_i$, and similarly the new first layer on the $T'$-side has fan-out $\tilde\Theta(\mathbf{d})$ for $\mathbf{d}=\prod_{j=1}^{h_0+1}d_j$. Moreover, by definition of $g_0$ and $h_0$, we have $p_{g_0+1},q_{h_0+1}\leq 1$, and by Observation~\ref{obs:naive-bound}, we have $p_{g_0+1}\geq \prod_{i=1}^{g_0+1}(\hat b_i/b_i)\geq 1/\mathbf{b}$, and $q_{h_0+1}\geq \prod_{j=1}^{h_0+1}(\hat d_j/d_j)\geq 1/\mathbf{d}$.
    
    It remains therefore only to analyze the number of connected terminal pairs (guaranteed with high probability) and expected solution cost specified in Lemma~\ref{lem:subsampling}. Let us start with the number of connected terminal pairs. Up to polylogarithmic factors, the guaranteed lower bound is
    \begin{align*}
        n_SD_S\cdot\prod_{i=1}^g p_i\cdot\prod_{j=1}^h q_j&=n_SD_S\cdot\frac{\hat D_S}{D_S}\cdot\frac{n_T}{\hat n_T}\cdot\prod_{i=1}^{g}\frac{\hat b_i}{b_i}\cdot\prod_{j=1}^{h}\frac{\hat d_j}{d_j}\\
        &=n_SD_S\cdot\frac{\hat D_S}{D_S}\cdot\frac{n_T}{\hat n_T}\cdot \frac{\hat n_S}{n_S}\cdot\frac{\hat n_T}{n_T}\\
        &=\hat D_S\hat n_S\\
        &=\tilde\Omega({\hat k}).
    \end{align*}

    Since we connect $\tilde\Omega({\hat k})$ demand pairs, 
    it remains to show that the solution cost is at most $\opt\cdot\tilde O\left(\sqrt{n_T\hat D_S/(\hat n_T D_S)}\right)$. We have already shown this for $S'$-side layers $i\leq g_0$ and $T'$-side layers $j\leq h_0$. For any layer $g'\geq g_0+1$, up to polylogarithmic factors,  Lemma~\ref{lem:subsampling} gives us cost bound
    \begin{align*}
        \opt\cdot\prod_{i=1}^{g'}\frac{p_ib_i}{\hat b_i}
        &=\opt\cdot\prod_{i=1}^{g_0}\frac{b_i}{\hat b_i}\cdot p_{g_0+1}\cdot\frac{b_{g_0+1}}{\hat b_{g_0+1}}\cdot\prod_{i=g_0+2}^{g'}\frac{(\hat b_i/b_i)b_i}{\hat b_i}\\
        &=\opt\cdot p_{g_0+1}\cdot \prod_{i=1}^{g_0+1}\frac{b_i}{\hat b_i}\\
        &=\opt\cdot \sqrt{\frac{\hat D_S}{D_S}\cdot\frac{n_T}{\hat n_T}}\cdot\prod_{i=1}^{g_0+1}\frac{\hat b_i}{b_i}\cdot \prod_{i=1}^{g_0+1}\frac{b_i}{\hat b_i}\\
        &=\opt\cdot\sqrt{\frac{\hat D_S}{D_S}\cdot\frac{n_T}{\hat n_T}},
    \end{align*}
    and the bound on the $T'$-side follows similarly, proving our approximation guarantee.
\end{proof}


\section{Improving the approximation}\label{sec:improve}

We introduce an additional idea to push the overall approximation guarantee down from $n^{1/3+\eps}$ to $n^{0.3+\eps}$ (or from $\tilde O(n^{1/3})$ to $\tilde O(n^{0.3})$ in our current setting). The additional idea is an algorithm that buys the entire $P'$-neighborhood $\Gamma_{P'}(t)\subseteq S'$ of one terminal $t\in T$, and then looks for a cheap tree on the $T$-side which connects many terminals in $T$ to the terminals in this neighborhood. Its guarantee is summarized as follows:

\begin{lemma}
    In our setting, there is an algorithm that produces a subtree with the following guarantees:
    \begin{enumerate}
        \item The solution connects at least $\tilde\Omega(\hat D_S\hat D_T)$ demand pairs.
        \item The cost of the solution is at most $\opt\cdot\tilde O(D_T/\hat D_T)$.
    \end{enumerate}
    In particular, this gives an $\tilde O((D_T/\hat D_T)\cdot {\hat k}/(\hat D_S\hat D_T))=\tilde O(\hat n_TD_T/(\hat D_S\hat D_T))$-approximation for MR$k$T.
    \label{lem:2-step}
\end{lemma}
\begin{proof}
    We guess some terminal $t\in T'$ that participates in our regularized optimum solution (by trying all possible $t\in T'$). By Lemma~\ref{lem:reg}, the $t$-constrained subtree, that is, the subtree spanning $\{r\}\cup\Gamma_{P'}(t)$ has quasi-regular fan-out $\tilde b_i$ in layer $i$, and the $t$-constrained subtree in the optimum solution has quasi-regular fan-out $\doub{b}_i$ in layer $i$ such that $\doub{b}_i\leq\min\{\hat b_i,\tilde b_i\}$,\footnote{Note that this is saying that the fan-out in layer $i$ of the $t$-constrained optimum sub-tree is bounded by both the corresponding fan-out in the full optimum tree and by the fan-out in the entire $t$-constrained subtree.} and $D_T=\prod_{i=1}^{g}\tilde b_i$ and $\hat D_T=\prod_{i=1}^g\doub{b}_i$. Thus, up to polylogarithmic factors, the entire $t$-constrained subtree 
    has weight at most
    $$\max_{g'\in[g]}\opt\cdot\prod_{i=1}^{g'}\frac{\tilde b_i}{\hat b_i}\leq \max_{g'\in[g]}\opt\cdot\prod_{i=1}^{g'}\frac{\tilde b_i}{\doub{b}_i}=\opt\cdot\prod_{i=1}^{g}\frac{\tilde b_i}{\doub{b}_i}=\opt\cdot\tilde\Theta(1)\cdot\frac{D_T}{\hat D_T},$$ which is within the claimed upper bound on the cost.

    For a correctly guessed $t$ that participates in the regularized optimum solution, we know that $\Gamma_{P'}(t)\subseteq S'$ must contain at least $\tilde\Omega(\hat D_T)$ terminals that neighbor $t$ in the optimum solution, and thus we know that there is a subtree of the tree spanning $\{r\}\cup T'$ of cost at most $\opt$ whose leaves form at least $\Omega(\hat D_S \hat D_T)$ pairs in $P'$ when coupled with $\Gamma_{P'}(t)$. The additional cost is even less than the required upper bound, and the number of demand pairs satisfies the required lower bound. We can find such a tree using the algorithm of Lemma~\ref{lem:one-side} (with sides reversed, first fixing terminal set $S_0=\Gamma_{P'}(t)$), and so the lemma follows.
\end{proof}

Our final guarantee follows by trading this guarantee against that of our main algorithm, as well as previous algorithms.

\begin{theorem}
    In our setting, there is an $\tilde O(n^{3/10})$-algorithm for MR$k$T.
    \label{thm:main-technical}
\end{theorem}
\begin{proof}
    We split the algorithm into two cases:

\paragraph{Case 1: $D_T/\hat D_T\geq n_S^{2/5}/n_T^{3/10}$.} In this case, equivalently rewriting the guarantee of Lemma~\ref{lem:sqrt} gives us an approximation ratio of (up to polylogarithmic factors) at most $$\sqrt{\frac{\hat D_T}{D_T}\cdot\frac{n_S}{\hat n_S}}\leq
\sqrt{\frac{n_T^{3/10}}{n_S^{2/5}}\cdot \frac{n_S}{\hat n_S}}=\sqrt{\frac{n_S^{3/5}n_T^{3/10}}{\hat n_S}},$$ and trading this off with the $\tilde O(\hat n_S)$ guarantee of Corollary~\ref{cor:hat-n} gives a ratio of $$\tilde O(n_S^{1/5}n_T^{1/10}),$$ which gives the approximation guarantee for this case.

\paragraph{Case 2: $D_T/\hat D_T\leq n_S^{2/5}/n_T^{3/10}$.} Consider the algorithm of Lemma~\ref{lem:2-step}. In this case, up to polylogarithmic factors, it gives an approximation guarantee of
$$\frac{\hat n_T}{\hat D_S}\cdot\frac{D_T}{\hat D_T}\leq \frac{\hat n_T}{\hat D_S}\cdot \frac{n_S^{2/5}}{n_T^{3/10}}.$$

This in itself doesn't give us a good approximation, but we combine it with two more algorithms. First, note that we can bound the approximation guarantee of Lemma~\ref{lem:sqrt} by $$\tilde O\left(\min\left\{\sqrt{n_T/\hat n_T},\sqrt{n_S/\hat n_S}\right\}\right)\leq\tilde O\left(\left(\frac{n_Sn_T}{\hat n_S\hat n_T}\right)^{1/4}\right).$$ Finally, our regularized instance (from Lemma~\ref{lem:reg}) is in particular a Steiner $k$-Forest instance with parameter $\tilde\Theta(\hat k)$, so the $O(\sqrt{k})$-approximation of~\cite{GHNR10} here gives us an $\tilde O\left(\sqrt {\hat k}\right)=\tilde O\left(\sqrt{\hat n_T\hat D_T}\right)$-approximation for MR$k$T.\footnote{Technically, this could return a forest rather than a tree, but that neither hurts the analysis nor is necessary -- we can always take the connected component with the best ratio in the solution.} Combining these three guarantees, up to polylogarithmic factors, we can bound our approximation guarantee by
\begin{align*}
    \min\left\{\frac{\hat n_S}{\hat D_T}\cdot \frac{n_S^{2/5}}{n_T^{3/10}},   \left(\frac{n_Sn_T}{\hat n_S\hat n_T}\right)^{1/4},   \sqrt{\hat n_T\hat D_T}\right\}&\leq \left(\frac{\hat n_S}{\hat D_T}\cdot \frac{n_S^{2/5}}{n_T^{3/10}}\right)^{1/7}\left(\left(\frac{n_Sn_T}{\hat n_S\hat n_T}\right)^{1/4}\right)^{4/7}\left(\sqrt{\hat n_T\hat D_T}\right)^{2/7}\\
    &=\left(\frac{n_S^{2/5}}{n_T^{3/10}}\cdot n_S n_T\right)^{1/7}\\
    &=n_S^{1/5}n_T^{1/10},
\end{align*}
which gives the required guarantee in the second case, and proves the theorem.
\end{proof}

This, in turn, gives us our main guarantee for Steiner $k$-Forest.
\begin{proof}[Proof of Theorem~\ref{thm:main}:]
    By the discussion in Section~\ref{sec:preliminaries}, an $\tilde O(n^{\alpha})$-approximation for MR$k$T in tree graphs translates to an $\tilde O(n^{\alpha})$-approximation for Steiner $k$-Forest in general graphs via the FRT~\cite{FRT04} embedding and the Steiner $k$-Forest to MR$k$T reduction we described. Moreover, thanks to Lemma~\ref{lem:CEK}, we may assume the input tree has bounded height $O(1/\eps)$ in exchange for an additional $n^{\eps}$-factor loss in the approximation guarantee. The additional polylogarithmic losses stem mainly from Lemma~\ref{lem:reg}, which, along with the various probabilistic arguments, brings the polylogarithmic-factor loss up to some $\log^{f(1/\eps)}n$. Combining this with a correct guess of parameters, and the correct bucket $P'$ of pairs containing $\tilde\Omega (\hat k)$ terminal-pairs connected in the optimum solution, Theorem~\ref{thm:main-technical} therefore guarantees an $O\left(n^{3/10+\eps}\cdot\log^{f(1/\eps)}n\right)$-approximation for Steiner $k$-Forest. The additional polylogarithmic factors can be absorbed by the polynomial factor by, say, choosing $\eps/2$ in place of~$\eps$, which proves the theorem.
\end{proof}

\section{Conditional hardness as a function of $k$}\label{sec:hard}

Hajiyaghayi and Jain~\cite{HJ06} give an approximation preserving reduction from Smallest $k$-Edge Subgraph (which they call Minimum $k$-Edge Coverage) to Steiner $k$-Forest (Theorem 6.2 of their paper). In fact, the reduction is quite simple: Given a Smallest $k$-Edge Subgraph instance $\langle G=(V,E),k\rangle$, let $H$ be a unit weight star with root $r\not\in V$ and terminals/leaves $V$, and for every edge $(u,v)\in E$ add a terminal pair $(u,v)$. Clearly, there is a $k$-edge subgraph in $G$ on $\ell$ vertices iff there is a weight~$\ell$ subtree of $H$ connecting at least $k$ terminal pairs.

Recently, Jones et al.~\cite{DkSSoS} showed that in a random Erd\H{o}s-R\'{e}nyi graph, in certain parameter regimes, there is a degree $n^{\Omega(1)}$ Sum-of-Squares relaxation (i.e., a \emph{pseudo-expectation}) for a small dense subgraph which is considerably denser than anything such a graph could contain. It had already been conjectured earlier (the ``Dense versus Random" conjecture -- see, e.g.,~\cite{chlamtavc2017minimizing}) that it is hard to distinguish whether the input is a random graph or a graph containing a dense subgraph, for these same parameter regimes. More specifically, the SoS gaps\footnote{Note that their construction treats the vertex and edge counts symmetrically, so the solution can be interpreted as a relaxation for either Densest $k$-Subgraph or Smallest $k$-Edge Subgraph.} of~\cite{DkSSoS} show that, w.h.p., the Erd\H{o}s-R\'{e}nyi graph $G(n,n^{-\beta})$ admits an $n^{\Omega(1)}$-degree pseudo-expectation (falsely) certifying the existence of an $\ell=n^{\alpha}$ subgraph with $k=n^{2\alpha-\gamma}$ edges, where $\alpha\in(0,1/2],\beta,\gamma\in(0,1)$ are constants (independent of the graph size) which satisfy
$$\gamma>\alpha\beta.$$ 
More specifically, by the analysis in~\cite{DkSSoS}, this holds for a pseudo-expectation of degree $n^{\Omega(\gamma-\alpha\beta)}$~\cite{potechin}.

To see how this implies conditional hardness for Steiner $k$-Forest, for any sufficiently small constant $\eps>0$, consider the following parameter setting:
$$\alpha=\eps/2,\qquad\beta=\eps,\qquad\text{and}\qquad\gamma=\eps^2/(1+\eps).$$
For these parameters, we have conditional hardness of distinguishing $G(n,n^{-\eps})$ from a graph with an $\ell=n^{\eps/2}$-subgraph containing at least $k=n^{\eps-\eps^2/(1+\eps)}$ edges. It is easy to check that, for these parameters, in $G(n,n^{-\eps})$, with high probability we need choose at least $\tilde\Omega(k)$ vertices to get $k$ edges. Thus, we get a gap of roughly
$$\frac{k}{\ell}=\frac{k}{n^{\eps/2}}=k^{1-(\eps/2)/(\eps-\eps^2/(1+\eps))}=k^{(1-\eps)/2},$$ for Smallest $k$-Edge Subgraph, which, again, is a special case of Steiner $k$-Forest. This rules out an $O(k^{1/2-\eps})$-approximation for either problem unless we can do better than degree $n^{\Omega(\gamma-\alpha\beta)}=n^{\Omega(\eps^2)}$ SoS integrality gaps for Densest $k$-Subgraph and refute the aforementioned Dense-versus-Random conjecture.

\section{Statement on extent of large language model use}

All the algorithms and proofs for the tree-metric setting were developed by the author. No algorithms, theorems, lemmas, or proofs presented here were generated by LLMs.

Other than suggesting the Steiner $k$-Forest problem and assisting in finding related work, the only research idea provided by an LLM that influenced this work was the use of tree embeddings itself, which came up in the following interaction: when asked to list all known combinatorial algorithms for the (closely related, but significantly easier) $k$-MST problem, ChatGPT mentioned the possibility of using tree embeddings such as FRT~\cite{FRT04}, which gave the author the idea to do the same thing for Steiner $k$-Forest. The author guessed that there should be a height reduction result with the same height/distortion trade-off as Zelikovsky~\cite{zelikovsky}, but in a form more directly applicable to our setting, and when asked specifically about this, ChatGPT supplied the reference~\cite{CEK06} as the source for such an embedding, which is quoted in Lemma~\ref{lem:CEK}.

Beyond that point, we gave ChatGPT explicit instructions not to supply any research ideas, and further interactions were limited to identifying typos or mistakes, consulting on wording and paper organization, and most crucially, workshopping the last sentence of the abstract to death.

\bibliographystyle{abbrvnat}
\bibliography{skf}

\appendix

\section{Proof of Lemma~\ref{lem:reg}}\label{app:reg}

We restate the lemma here for convenience:

\begin{restatement}{Lemma}{lem:reg}
    
    Given a bounded-diameter $n$-vertex tree $G=(V,E)$ with positive edge weights $w$ in some range $[w_0,n\cdot w_0]$, a root $r\in V$, and a set of terminal pairs $P\subseteq S\times T$ where $S,T\subseteq V$, 
    we can efficiently find terminal subsets $S'\subseteq S$ and $T'\subseteq T$ and a set of pairs $P'\subseteq P\cap(S'\times T')$ of size $|P'|=\tilde\Omega(|P|)$, with the following properties:
    \begin{itemize}
    \item For some $c=O(1)$ that depends only on the diameter of $G$, there exist $D_S,D_T\in\nats$ such that every terminal $s\in S'$ has terminal-pair degree in the range $D_S/\log^c n\leq\deg_{P'}(s)\leq D_S$, and every terminal $t\in T'$ has degree in the range $D_T/\log^c n\leq\deg_{P'}(t)\leq D_T$.    
    \item There exist $g,h\in\nats$ such that all terminals in $S'$ are at depth $g$ and all terminals in $T'$ are at depth $h$ (note that these are constants bounded by the diameter of $G$).
    \item There exist $b_1,\ldots,b_g\in\nats$ such that, in the rooted subtree spanning $\{r\}\cup S'$, at each layer $i\in[g]$, every parent node has fan-out in the range $[b_i/\log^c n,b_i]$. Moreover, every edge in this layer has weight in the range $[\beta_i,2\beta_i]$ for some $\beta_i> 0$.
    \item There exist $d_1,\ldots,d_h\in\nats$ such that, in the rooted subtree spanning $\{r\}\cup T'$, at each layer $j\in[h]$, every parent node has fan-out in the range $[d_j/\log^c n,d_j]$. Moreover, every edge in this layer has weight in the range $[\delta_j,2\delta_j]$ for some $\delta_j> 0$.
    \item There exist $\tilde b_1,\ldots,\tilde b_g$ such that for every terminal $t\in T'$, in the rooted subtree spanning $\{r\}\cup \Gamma_{P'}(t)$, at each layer $i$ (of the $S'$-side of the tree), every parent node has fan-out in the range $[\tilde b_i/\log^c n,\tilde b_i]$.
    \item There exist $\tilde d_1,\ldots,\tilde d_h$ such that for every terminal $s\in S'$, in the rooted subtree spanning $\{r\}\cup \Gamma_{P'}(s)$, at each layer $j$ (of the $T'$-side of the tree), every parent node has fan-out in the range $[\tilde d_j/\log^c n,\tilde d_j]$.
\end{itemize}    
\end{restatement}

The construction relies heavily on the following result (slightly simplified from Lemma 2.6 of~\cite{Chlamtac-Dinitz}):

\begin{lemma}
    For any constant $c_1\in\nats$, there exists a constant $c_2\in\nats$ such that, given an $n$-vertex graph $G=(V,E)$, vertices $u,v$, and a set of paths $P$ of length (number of edges) exactly $c_1$ between $u$ and $v$, we can efficiently find a subset $P'\subseteq P$ with the following properties:
    \begin{itemize}
        \item There exist pairwise disjoint edge sets $E_1,\ldots, E_{c_1}\subset E$ such that $P'\subseteq\prod_{i=1}^{c_1}E_i$. That is, the paths in $P'$ form a layered graph.
        \item $|P'|\geq|P|/\log^{c_2}n$.
        \item There exists a positive integer vector $(m_I)_{I\subseteq[c_1]}$ such that for every non-empty index set $\emptyset\neq I\subseteq[c_1]$, for every $I$-tuple of edges $(e_i)_{i\in I}\in\prod_{i\in I}E_i$, either no path in $P'$ contains all the edges in this tuple, or the number of paths that contain them is in the range $$|\{p\in P'\mid p|_I=(e_i)_{i\in I}\}|\in[m_I,m_I\log^{c_2}n].$$
    \end{itemize}
    \label{lem:chlamtac-dinitz}
\end{lemma}

\begin{remark}
    The original lemma statement in~\cite{Chlamtac-Dinitz} did not explicitly guarantee an efficient construction, since in that context only existence was needed for the analysis. However, polynomial time follows easily from their proof, since the regularizing procedure involved repeatedly finding a constant size edge set with a ``degree" violation, and deleting a non-empty set of paths from $P$.
\end{remark}

We can now prove Lemma~\ref{lem:reg}.

\begin{proof}
    Denote by $L(i)$ the set of depth $i$ vertices (in our $r$ rooted tree), and for every pair $(s,t)\in P$, let $p_{s,t}$ be the concatenation of the simple path from $s$ to $r$ in $G$ with the simple path from $r$ to $t$. Since the diameter of $G$ is bounded, there are some $g,h$ such that the set of pairs $$P_0=\{(s,t)\in P\mid s\in L(g), t\in L(h)\}$$ has cardinality $|P_0|=\Omega(|P|)$. We can also bucket all the edges into $\log n$ buckets by weight, so that, letting $E^{(i)}=\{e\in E\mid w(e)\in[w_0 2^{i-1},w_02^i)\}$, there is some vector $(\gamma_j)_{j\in[g+h]}\in [\log n]^{g+h}$, such that the set of pairs
    $$P_1=\left\{(s,t)\in P_0\;\left|\; p_{s,t}=(e_1,\ldots,e_{g+h})\text{, where } \forall j\in[g+h]:e_j\in E^{(\gamma_j)}\right.\right\}$$ has cardinality
    $$|P_1|\geq |P_0|/\log^{g+h}n.$$
    This ensures both the uniformity in depth of $S$ terminals and in depth of $T$ terminals, as well as the approximate (up to factor 2) uniformity in edge weights at every level of the paths. Finally, we double the graph, joining the two copies at $r$, identifying the $S$ terminals with one copy, and the $T$ terminals with the other, so that for every $(s,t)\in P_1$, the $r$ to $t$ portion of $p_{s,t}$ comes from a distinct copy from the $s$ to $r$ portion (other than $r$ itself), thus making the paths simple.

    The rest of the properties will follow once we apply Lemma~\ref{lem:chlamtac-dinitz}. To do so, we create a set of paths as follows: create new vertices $u,v$, connect $u$ to all terminals in $S$, and $v$ to all terminals in~$T$. For every pair $(s,t)\in P_1$, create a new path $(u,s)\circ p_{s,t}\circ(t,v)$. Let $P_2$ be the resulting path set. Note that $|P_2|=|P_1|$ and that every path in $P_2$ has unweighted length exactly $g+h+2$. Now apply Lemma~\ref{lem:chlamtac-dinitz} to $P_2$, and let $P'_2$ be the resulting set of paths, with parameters $(m_I)_{I\subseteq[g+h+2]}$ as in the lemma. Create $P'$ by deleting from these paths the initial and final edges incident on $u$ and $v$, respectively, and let $S',T'$ be the sets of all initial and final vertices in these paths, respectively.

    For convenience, for index sets $I\subseteq[g+1]$ and $J\subseteq[h+1]$ corresponding to the layers of the $S'$-side and $T'$-side of the tree\footnote{Intuitively, think of layers $g+1$ and $h+1$ as representing the terminals.}, respectively, we translate the $(m_I)$ vector to a vector indexed by the layers of our two-sided tree in the natural way:
    $$m_{I}^{J}:=m_{\{g+2-i\,|\, i\in I\}\cup\{j+g+1\,|\,j\in J\}}.$$
    It is easy to see that every degree parameter\footnote{Note that $b_1,d_1$ are simply root-degrees in the tree and do not require a separate uniformity bound.} defined by Lemma~\ref{lem:reg} is a function of the vector $\left(m_I^J\right)_{I\subseteq[g+1]}^{J\subseteq[h+1]}$, with the approximate regularity condition following from the corresponding regularity condition guaranteed by Lemma~\ref{lem:chlamtac-dinitz}. In particular, we have: 
\begin{align*}
    D_S=\tilde\Theta\left(m_{\{g+1\}}^{\emptyset}\right) \qquad&\text{and}\qquad D_T=\tilde\Theta\left(m_{\emptyset}^{\{h+1\}}\right)\\
    b_i=\tilde\Theta\left(m_{\{i-1\}}^{\emptyset}/m_{\{i-1,i\}}^{\emptyset}\right)\qquad&\text{and}\qquad d_j=\tilde\Theta\left(m^{\{j-1\}}_{\emptyset}/m^{\{j-1,j\}}_{\emptyset}\right)&\forall i\in\{2,\ldots,g\},\forall j\in\{2,\ldots, h\}\\
    \tilde b_1=\tilde\Theta\left(m_{\emptyset}^{\{h+1\}}/m_{\{1\}}^{\{h+1\}}\right) \qquad&\text{and}\qquad \tilde d_1=\tilde\Theta\left(m^{\emptyset}_{\{g+1\}}/m^{\{1\}}_{\{g+1\}}\right)\\
    \tilde b_i=\tilde\Theta\left(m_{\{i-1\}}^{\{h+1\}}/m_{\{i-1,i\}}^{\{h+1\}}\right) \qquad&\text{and}\qquad \tilde d_j=\tilde\Theta\left(m^{\{j-1\}}_{\{g+1\}}/m^{\{j-1,j\}}_{\{g+1\}}\right)&\forall i\in\{2,\ldots,g\},\forall j\in\{2,\ldots, h\}
\end{align*}
\end{proof}

\end{document}